\documentclass[a4paper, 11 pt]{article}
\usepackage [a4paper,left=3.1 cm,bottom=2.7cm,right=3.1 cm,top=2.7cm]{geometry}
\usepackage[english]{babel}
\usepackage{amssymb}
\usepackage{amsmath,amsthm}
\usepackage[normalem]{ulem}
\usepackage{subcaption}
\usepackage{tabularx,array}
\usepackage{graphicx}
\usepackage{soul} 
\usepackage{enumitem} 
\usepackage{mathtools}
\usepackage{xcolor}
\def\E{\mathbb{E}}
\def\P{\mathbb{P}}

\newtheorem{remark}{Remark}
\newtheorem{theorem}{Theorem}
\newtheorem{lemma}{Lemma}

\newtheorem{h}{Hypothesis}

\newcommand{\keywords}[1]{\par\addvspace\baselineskip
\noindent\enspace\ignorespaces#1}

\begin{document}

\title{On the implied volatility of Asian options under stochastic volatility models}

\author{Elisa Al\`os$^{\ast}$, Eulalia Nualart$^{\ast}$ and Makar Pravosud\thanks{ Universitat Pompeu Fabra and Barcelona School of Economics, Department of Economics and Business, Ram\'on Trias Fargas 25-27, 08005, Barcelona, Spain. 
EN acknowledges support from the Spanish MINECO grant PGC2018-101643-B-I00 and
Ayudas Fundacion BBVA a Equipos de Investigaci\'on Cient\'ifica 2017.}}
\maketitle

\begin{abstract}
 In this paper we study the short-time behavior of the at-the-money implied volatility for arithmetic Asian options with fixed strike price.
 The asset price is assumed to follow the Black-Scholes model with a general stochastic volatility process. Using techniques of the Malliavin calculus developed in Al\`os, Garc\'ia-Lorite and Muguruza \cite{Alos2018} we give sufficient conditions on the stochastic volatility in order to compute the level of the implied volatility of the option when the maturity converges to zero. Then, we find a short maturity asymptotic  formula for the skew slope of the implied volatility that depends on the correlation between prices and volatilities and the Hurst parameter of the volatility model. We apply our general results to the SABR and fractional Bergomi models, and provide numerical simulations that confirm the accurateness of the asymptotic formulas. 
\end{abstract}


\section{Introduction}\label{introduction}
This paper is devoted to the study of Asian call options with payoff of the form
$$
\left(\frac{1}{T}\int_0^T S_udu-K\right)_+,
$$
where $T$ denotes the maturity, $S$ the price of the underlying, and $K$ the strike of the contract. Asian options of this type are  extremely important in energy markets for different reasons. From one hand, typical energy transactions use to take place via multiple deliveries. Then, these transactions are priced on average and not only on a terminal price. Secondly, the payoff is less sensitive to extreme market fluctuations, which becomes interesting in non-liquid markets. Finally, these options tend to be cheaper than the corresponding European vanillas.

The aim of this paper is to study the short-time maturity behavior of the at-the-money implied volatility (ATMIV) of
arithmetic Asian options. The study of implied volatility is useful in many ways. Firstly, it can be used to obtain volatilities for pricing OTC options (and other derivatives) with  strikes and maturities that are different from the ones offered by option exchanges. Secondly, the shape of the implied volatility surface can be used to assess the adequacy of an option pricing model. If the option pricing model is adequate, then it should capture the main properties of the empirical implied volatility surface. In particular, one of the key characteristics of the implied volatility is its skew at the short end and one can easily filter the class of suitable models if the theoretical value of the skew is available for the models of interest. Finally, as one will see further in the paper, one can use implied volatility and its skew to efficiently approximate the option price. Last, but not least, due to the smile effect the hedge ratio has to be adjusted to take into account the market skew. As a result, availability of analytical values of the skew can improve the performance of hedging.

The behavior of the implied volatility for vanilla options has been the object of many works (see for example Lee \cite{L} for a basic introduction to this topic). However, the case of exotic options and more specifically Asian options is less studied and the number of exact analytic results is more limited. 

Yang and Ewald \cite{Yang09} compute the implied volatility for OTC traded Asian options under Black-Scholes with constant volatility by combining Monte-Carlo techniques with the Newton method in order to solve nonlinear equations. Approximation methods for pricing Asian options under stochastic volatility models are studied by Forde and Jacquier \cite{Forde2010}.
Chatterjee et al. \cite{Chat2018} develop a Markov chain-based approximation method to price arithmetic Asian options for short maturities under the case of geometric Brownian motion. Fouque and Han \cite{Fouque2003} generalize the dimension reduction technique of Vecer for pricing arithmetic Asian options. They use the fast mean-reverting stochastic volatility asymptotic analysis to derive an approximation to the option price which takes into account the skew of the implied volatility surface. This approximation is obtained by solving a pair of one-dimensional partial differential equations. The methodology requires the key parameters needed in the PDEs to be estimated from the historical stock prices and the implied volatility surface.

Asymptotics of arithmetic Asian implied volatilities have been studied by Pirjol and Zhu \cite{Pirjol2016} in the case of local volatility. In this paper, the authors make use of large deviations techniques to get accurate approximation formulas for the implied volatility, which are shown to be accurate when compared with the Monte-Carlo simulations.
Arithmetic Asian options under the CEV (constant elasticity of variance) model are studied in  Pirjol and Zhu \cite{Pirjol2019}. The leading order short maturity limit of the Asian option prices under the CEV model is obtained in closed form. Authors propose an analytic approximation for the Asian options prices which reproduced the exact short maturity asymptotics. 
Al\`os and Le\'on \cite{Alos2019} compute the short-time level and the skew of the implied volatility of floating strike arithmetic Asian options under the Black-Scholes model with constant volatility by the means of Malliavin calculus.

In this paper, we contribute to the existing literature in several ways. We extend the application of the Malliavin calculus developed in Al\`os, Garc\'ia-Lorite and Muguruza 
\cite{Alos2018} by giving general sufficient conditions on a general stochastic volatility model in order to obtain formulas for the short-time limit of the  at-the-money level and skew of the implied volatility for Asian options. Moreover, we show how studying Asian option under stochastic volatility reduces to the study of European type options where the underlying is respresented by a certain stochastic volatility model, with a modified volatility process which depends on \(T\). This methodology developed in 
\cite{Alos2018} allows to adapt the results on vanilla options to options on a non lognormal-type distribution and it can also be applied to other European-type exotic options. See for example \cite{Alos2018} for an application of this technique to the analysis of the VIX skew.  This method is very classical in mathematical finance, for instance, when working with stochastic rates. The difference in the present paper is that the volatility of the log-forward price (see equation (\ref{x})) depends on the filtration of the Brownian motion $W$, and not only on $W'$. This makes the problem more challenging in order to find general hypotheses on a general stochastic volatility model in order for the asymptotic formulas to hold (see Hypotheses 3 and 4). Typically these hypotheses are in expectation, see for example Al\`os and Shiraya \cite{Alos2019b} and Al\`os and Le\'on \cite{Alos2017}, but in our case they need to hold almost surely in order to check the hypotheses of Theorem 3 (see Section 4.2).

To sum up, we study the short-end behavior of the ATMIV of Asian options for local, stochastic, and fractional volatilities. In particular, we show that
\begin{itemize}
\item The short-end limit of the ATMIV is equal to $\frac{\sigma_0}{\sqrt{3}}$, where $\sigma_0$ denotes the short-end limit of the spot volatility. See equation (\ref{main1}) in Theorem \ref{limskew}.

\item  We compute the the short-end skew of the ATMIV, which depends on the correlation between prices and stochastic volatility and on the Malliavin derivative of the volatility process, which in the case of fractional volatility models will depend on the Hurst parameter $H \in (0,1)$. See equation (\ref{main2}) in Theorem \ref{limskew}. If prices and volatilities are uncorrelated, the short-end skew is equal to $\frac{\sqrt{3}\sigma_0}{30}$. In the case of rough volatilities, that is $H<\frac12$, we observe in equations (\ref{rou}) and (\ref{rou2}) a blow-up that is of the same order as the one we observe in vanilla options (see Al\`os et al. \cite{Alos2007}). 

\item We apply the preceding results to the constant volatility case,  the SABR model,  the fractional Bergomi model, and the local volatility, and perform numerical simulations that confirm the accurateness of the asymptotic formulas. See Section 5. In the case of local volatilities, we verify that our results fit the asymptotic analysis of Pirjol and Zhu \cite{Pirjol2016} and \cite{Pirjol2019}.
\end{itemize}

The paper is organized as follows: in Section 2 we introduce the main elements of the Malliavin calculus needed through the paper, and the main problem, results and notations. In Section 3 we introduce some preliminary results needed for the proof of the  main theorem. In Section 4 we give the proof the main results of the paper. Finally, Section 5 is devoted to the application of the main results to the constant volatility case, the SABR model, the fractional Bergomi model, and the local volatility model, together with some numerical simulations to confirm the accurateness of the asymptotic formulas. The Appendix contains some Malliavin derivatives computations needed through the paper.

\section{Notations and main results}

\subsection{A primer on Malliavin Calculus}

We introduce the elementary notions of the
Malliavin calculus used in this paper (see Nualart and Nualart \cite{Nua18}). Let us consider a standard Brownian motion $Z=(Z_t)_{t \in [0,T]}$ defined on a complete probability space $(\Omega, \mathcal{F}, \mathbb{P})$ and we denote by $\mathcal{F}_t$ the filtration generated by $Z_t$. Let ${\cal S}^Z$ be the set of random variables of the form
\begin{equation}
\label{eq:2.1}
F=f(Z(h_{1}),\ldots ,Z(h_{n})),  
\end{equation}
with $h_{1},\ldots ,h_{n}\in L^2([0,T])$, $Z(h_i)$ denotes the Wiener integral of the function $h_i$, for $i=1,..,n$, and $f\in C_{b}^{\infty }(\mathbb{R}^n) $ 
(i.e., $f$ and all its partial derivatives are bounded). Then the Malliavin 
derivative of $F$, $D^Z F$,  is defined
as the stochastic process given by 
\begin{equation*}
D_{s}^ZF=\sum_{j=1}^{n}{\frac{\partial f}{\partial x_{j}}}(Z(h_{1}),\ldots ,Z(h_{n})) h_j(s), \quad s\in [0,T].
\end{equation*}
This operator is closable from $L^{p}(\Omega )$ to $L^p(\Omega; L^2([0,T])$, for all $p \geq 1$, and we denote by ${\mathbb{D}}_{Z}^{1,p}$ the
closure of ${\cal S}^Z$ with respect to the norm
$$
||F||_{1,p}=\left( \E\left| F\right|
^{p}+\E||D^Z F||_{L^{2}([0,T])}^{p}\right) ^{1/p}.
$$
We also consider the iterated
derivatives \(D^{Z,n}\) for all integers  \(n > 1\) whose domains will be denoted by
\(\mathbb{D}^{n,p}_Z\), for all $p \geq 1$. We will use the notation $\mathbb{L}_Z^{n,p}:=L^p([0,T];{\mathbb{D}}_{Z}^{n,p})$.

\subsection{Statement of the problem and main results}\label{statement-of-the-problem-and-notation}

We denote by $(V_t)_{t \in [0,T]}$ the value of a fixed strike arithmetic Asian call option where $T$ is the maturity.  Then, the payoff can be written as 
$$
V_T=(A_T-K)_{+}, \qquad  A_T=\frac{1}{T}\int_0^T S_t dt,
$$
where $(S_t)_{t \in [0,T]}$ is
the price of the underlying asset  and \(K\) is the
fixed strike price.

For the sake of simplicity, we assume that the interest rate is equal to zero and we consider the following general stochastic volatility model for the underlying asset price
\begin{equation} \begin{split}\label{bspm}
dS_t &= \sigma_tS_tdW_t\\
 W_t &= \rho W_t' + \sqrt{(1-\rho^2)}B_t,
 \end{split}
\end{equation} where $S_0>0$ is fixed, $W_t$, $W_t'$, and $B_t$ are three standard Brownian motions on $[0,T]$
defined on the same complete probability space $(\Omega, \mathcal{G}, \mathbb{P})$. We assume that \(W_t'\) and
\(B_t\) are independent and \(\rho\in [-1,1]\) is the correlation coefficient between \(W_t\) and
\(W_t'\). 

We consider the following assumption on the stochastic volatility of the asset price.
\begin{h} \label{Hyp1}
The process $\sigma=(\sigma_t)_{t \in [0,T]}$ is adapted to the filtration generated by $W'$, a.s. positive and continuous, and satisfies that for all $t \in [0,T]$,
$$
c_1 \leq \sigma_t \leq c_2,
$$
for some positive constants $c_1$ and $c_2$.
\end{h}

\begin{remark}
Hypothesis \ref{Hyp1} may seem too restrictive since it is not satisfied by the stochastic volatility models considered in Section 5. However, we will show that using a truncation argument, Theorem \ref{limskew} is still true in all our examples.
\end{remark}

We define the forward price as the martingale $M_t=\E_t(A_T)$, where $\E_t$ denotes the conditional expectation wrt to the filtration $\mathcal{F}_t$ generated by $W_t$.  Applying the stochastic Fubini's theorem we get that
\begin{align*}
\begin{split}
  A_T=\frac{1}{T}\int_0^T S_t dt &= \frac{1}{T}\int_0^T \left(S_0+\int_0^t \sigma_u S_u dW_u\right) dt = \\
   &= S_0+\frac{1}{T}\int_0^T\sigma_u S_u \left(\int_u^T  dt\right) dW_u = \\
   &= S_0+\frac{1}{T}\int_0^T (T-u)\sigma_u S_u dW_u,
\end{split}
\end{align*}
which implies that 
\begin{equation}\label{dACE}
dM_t=\frac{\sigma_t S_t(T-t)}{T} dW_t = \phi_t M_t dW_t,
\end{equation} where $$\phi_t := \frac{\sigma_t S_t(T-t)}{T M_t}.$$ Furthermore, the log-forward price $X_t=\log(M_t)$ satisfies
\begin{equation} \label{x}
dX_t=\phi_t dW_t-\frac12 \phi_t^2 dt.
\end{equation}

\begin{remark} \label{r1}
One can easily check  that  Hypothesis \ref{Hyp1} implies that $\phi_t$ is positive a.s. and belongs to $L^p([0,T]\times \Omega)$, for all $p \geq 2$. In fact, Hypothesis \ref{Hyp1} implies that for all $p \geq  2$,
$S_t$ belongs to $L^p([0,T] \times \Omega)$, $A_T$ belongs to $L^p(\Omega)$, and $M_t^{-1}$ belongs to $L^p([0,T] \times \Omega)$.
\end{remark}

\begin{remark} \label{r11}
Notice that $\phi_0=\sigma_0$. Moreover, 
\begin{equation} \label{er}
M_t=\frac{1}{T}\left(\int_0^t S_udu+S_t(T-t)\right).
\end{equation}
Then, $TM_t\geq S_t(T-t)$, and this implies that $\phi_t\leq \sigma_t$ almost surely.
\end{remark}

The goal of this paper is to study the implied volatility of the Asian call option $V_t$ which is defined as follows. We denote by \(BS(t,x,k,\sigma)\)  the classical Black-Scholes price of a
European call with time to maturity \(T-t\), log-forward price \(x\),
log-strike price \(k\) and volatility \(\sigma\). That is, \begin{align*}
\begin{split}
  BS(t,x,k,\sigma)&=e^x N(d_{+}(k,\sigma))-e^k N(d_{-}(k,\sigma)),\\
  d_{\pm}(k,\sigma)&=\frac{x-k}{\sigma\sqrt{T-t}}\pm\frac{\sigma}{2}\sqrt{T-t},
\end{split}  
\end{align*} where \(N\) is the cumulative distribution
function of the standard normal random variable.

Next, we observe that, as $BS(T,x,k,\sigma)=(e^x-e^k)_+$ for every $\sigma>0$,  the price of  our Asian call option $V_t=\E_t(e^{X_T}-e^k)_+$ can be written as 
\begin{equation} \label{vt}
V_t=\mathbb{E}_t(BS(T,X_T,k,v_T)), \qquad v_t=\sqrt{\frac{1}{T-t}\int_t^T \phi_s^2ds}.
\end{equation}
 In particular,
\(V_T=BS(T,X_T,k,v_T)\).
Then, we define the implied volatility of the option as \(I(t,k)=BS^{-1}(t,X_t,k,V_t)\), and we denote by \(I(t,k^{*}_t)\), where $k^{\ast}_t=X_t$, the corresponding ATMIV which, in the case of zero interest rates, takes the form
\(BS^{-1}(t,X_t,X_t,V_t)\).

We apply  the Malliavin calculus techniques developed in Al\`os, Garc\'ia-Lorite and Muguruza 
\cite{Alos2018} in order to obtain formulas for  $$\lim_{T\to 0}I(0,k^{*}) \quad \text{and} \quad 
\lim_{T\to 0}\partial_kI(0,k^{*})$$ under the general stochastic volatility model (\ref{bspm}), where we have set $k^{\ast}=k_0^{\ast}$ for the sake of simplicity.

In our setting, since we have three  Brownian motions $W, W'$ and $B$, if $h$ is a random variable in $L^2([0,T])$, then we have in view of relation (\ref{bspm}) that
$$
W(h)=\rho W'(h)+\sqrt{(1-\rho^2)} B(h).
$$
Then, a random variable in $\mathbb{D}^{1,2}_{W'}\cap \mathbb{D}^{1,2}_B$ is also in $\mathbb{D}^{1,2}_{W}$. In fact, it is easy to see that if $X$ is a random variable in ${\cal S}^W$, then 
\begin{equation}
D^{W}X= \rho D^{W'}X+\sqrt{1-\rho^2}D^{B}X.
\end{equation}
Thus, we deduce that for all $X\in \mathbb{D}^{1,2}_{W'}\cap \mathbb{D}^{1,2}_B$,
 \begin{align}
  D^{W}X= \rho D^{W'}X+\sqrt{1-\rho^2}D^{B}X.
\label{derprod}
\end{align}

We will need the following additional assumption on the Malliavin differentiability of the stochastic volatility process.
\begin{h}\label{Hyp2}
$\sigma \in \mathbb{L}^{2,p}_{W'}$, for all $p \geq 2$.
\end{h}

\begin{remark} \label{r2}
Hypotheses 1 and 2 imply that $\phi_t$ belongs to $\mathbb{L}^{2,p}_{W}$ and $A_T$ belong to $\mathbb{D}^{2,p}_{W}$ for all $p \geq 2$. This  hypothesis on $A_T$ corresponds to {\bf (H1)} in Al\`os, Garc\'ia-Lorite and Muguruza \textnormal{\cite{Alos2018}}.
\end{remark}

In order to give the asymptotic skew of the implied volatility as a function of the roughness of the stochastic volatility process we consider the following assumption.
\begin{h}\label{Hyp3}
There exists  $H\in (0,1)$ such  that for all $0\leq s \leq r \leq t \leq T $ 
\begin{equation} \label{d1}
\vert D_r^{W'}\sigma_t \vert  \leq M_{r,t} (t-r)^{H-\frac{1}{2}}
\end{equation}
and
\begin{equation} \label{d2}
\vert D_s^{W'} D_r^{W'}\sigma_t \vert \leq N_{s,r,t} (t-r)^{H-\frac{1}{2}} (t-s)^{H-\frac{1}{2}},
\end{equation}
where  $M_{r,t}$ and $N_{s,r,t}$ are positive random variables satisfying for all $p\geq  1$,
$$
\E(\sup_{0\leq r \leq t \leq T \leq 1} M_{r,t}^p) \leq c_1,
$$
and
$$
 \E(\sup_{0\leq s \leq r \leq t \leq T \leq 1} N_{s,r,t}^p) \leq c_2,
$$
for some positive constants $c_1$ and $c_2$.
\end{h}

Finally, we will need the following additional assumption on the continuity of the paths of the volatility process.
\begin{h}\label{Hyp4}
There exists $\gamma \in (0, H)$ such that for all $0\leq s\leq r\leq T \leq 1$
$$|\sigma_r-\sigma_s| \leq  K_{r,s}(r-s)^{\gamma},$$
where $K_{r,s}$ is a positive random variable satisfying for all $p \geq 1$,
$$\E(\sup_{0\leq r \leq t \leq T \leq 1}   K_{r,s}^p) \leq c
$$
where $c>0$ and $H$ is the Hurst parameter form Hypothesis \textnormal{3}.
\end{h}

We next provide the the main result of this paper, which is the short-time ATMIV level and skew of an Asian call option under the general volatility model (\ref{bspm}).
\begin{theorem}  \label{limskew}
Assume Hypotheses 1-4. Then, 
\begin{align} \label{main1}
\lim_{T\to 0}I(0,k^*)=\frac{\sigma_0}{\sqrt{3}}
\end{align}
and
\begin{equation} \label{main2}\begin{split}
&\lim_{T \to 0} T^{\max(\frac12-H, 0)}\partial_kI(0,k^{*}) \\
 &=\lim_{T \to 0} T^{\max(\frac12-H, 0)}\frac{3 \sqrt{3}\rho}{ \sigma_0 T^5}\int_0^T\left( (T-r)\int_r^T(T-u)^2 \E(D_r^{W'}\sigma_u) du \right) dr \\
 &\qquad \qquad +\lim_{T \to 0} T^{\max(\frac12-H, 0)}\frac{\sqrt{3}\sigma_0 }{30},
 \end{split}
\end{equation}
and the limit on the right hand side of\textnormal{(\ref{main2})} is finite.
\end{theorem}

We observe that the level (\ref{main1}) is independent of the  correlation $\rho$ and the Hurst parameter $H$, and coincides with the constant volatility  case, see Pirjol and Zhu \cite{Pirjol2016} and Al\`os and Le\'on \cite{Alos2019}. Observe  also  that it coincides with the a.s. limit of $v_0$. In fact, by Hypothesis 1 and since $S_0=M_0$ we have that a.s.
\begin{equation} \label{limit}
\lim_{T \rightarrow 0} v_0=\lim_{T \rightarrow 0} \sqrt{\frac{1}{T^3}\int_0^T \frac{\sigma_s^2 S_s^2(T-s)^2}{ M_s^2} ds}=\frac{\sigma_0}{\sqrt{3}}.
\end{equation}

The skew (\ref{main2}) depends on the correlation parameter $\rho$ and on the Hurst parameter $H$. When prices and volatilities are uncorrelated  then the short-time skew equals
$\frac{\sqrt{3}\sigma_0 }{30}$, which again coincides with the constant volatility  case, see Pirjol and Zhu \cite{Pirjol2016} and Al\`os and Le\'on \cite{Alos2019}. Observe also that if the term $\E(D_r^{W'}\sigma_u)$ is of order $(u-r)^{H-\frac12}$ (see Hypothesis 3), then the  limit of the right hand side of (\ref{main2}) will be $0$ if $H>1/2$ and it will converge to a constant when $H=\frac12$. When $H <\frac12$ we need to multiply by $T^{\frac12-H}$ in order to obtain a finite limit.
 This is because when  $H>1/2$, the fractional Brownian motion is smoother than standard Brownian motion and the effect  of the stochastic volatility on the short-time implied volatility will be the  same as it was constant, while when  $H<1/2$, the fractional Brownian motion is rougher than standard Brownian motion and we obtain the same effect as in the  case of vanilla options, see Al\`os et al. \cite{Alos2007}.

The results of Theorem  \ref{limskew} can be used in order to derive an approximation formula for the price of an Asian call option. By definition, the price of the Asian call option writes as
$$
V_0 = BS(0,X_0,k,I(0,k))
$$
Then, using Taylor's formula we can use the approximation
\begin{equation} 
\label{IVProxy}
I(0,k) \approx I(0,k^{*}) + \partial_kI(0,k^{*})(k-k^{*}).
\end{equation}
Of course, this approximation is only linear and can be expected to have a limited validity, restricted to a narrow region around the ATM point. One would expect to obtain better results if one has a short maturity asymptotic formula for the curvature 
$\partial^2_k I(0,k^{\ast})$. The
short-time maturity asymptotics for the ATM curvature of the implied volatility
for European calls under general stochastic volatility models is computed in Al\`os and Le\'on \cite{Alos2017}. A second order Taylor
expansion for short maturity limit of the implied volatility for Asian options around the ATM point when the underlying asset
follows a local volatility model is obtained in Proposition 19 of Pirjol and Zhu \cite{Pirjol2016}. In the constant volatility case, the short maturity limit is known for all $k$ (see Proposition 8 of \cite{Pirjol2016}), and the short-maturity limit of the skew and the curvature for $k^{\ast}$ are known in closed from  Pirjol \cite{Pirjol2023}.
In our setting, computing the curvature and a second order approximation of the price around the ATM point
is more challenging and we leave it for further work. 

\section{Preliminary results}

We start quoting the  two main results obtained in  Al\`os et al. \cite{Alos2018} that will be  crucial for the proof of our main Theorem and use the general framework detailed in Section 2.2.
The first result is Theorem 6 in Al\`os et al. \cite{Alos2018} which shows that the short-time limit of the ATMIV equals the  short-time limit of  the future average of the volatility of the log forward price.
\begin{theorem} \label{theorem8A}
Assume that for all $p >1$, $A_T \in \mathbb{D}^{2,p}_W$, $M_t^{-1} \in L^p([0,T]\times \Omega)$, and
\begin{align}
\lim_{T\to 0}& \mathbb{E} \left(\int_0^T\frac{\Lambda_s}{v_s^2(T-s)}ds\right)  = 0,\label{j1}\\
\lim_{T\to 0}& \frac{1}{T^2}\mathbb{E}\left(\frac{1}{v_0} \int_0^T \left(\int_s^T D^W_s\phi_r^2dr\right)^2 ds \right) = 0,\label{j2}
\end{align}
where $\Lambda_s=\phi_s\int_s^T D_s^W\phi_r^2dr$.
Then,
$$
  \lim_{T\to 0} I(0,k^*)=\lim_{T\to 0}\mathbb{E}(v_0).
$$
\end{theorem}

The second result is Theorem 8 of Al\`os et al. \cite{Alos2018} which gives an approximation formula for the short-time limit of the ATMIV skew.
\begin{theorem} \label{theorem8B}
Assume that for all $p >1$, $A_T \in \mathbb{D}^{3,p}_W$, $M_t^{-1} \in L^p([0,T]\times \Omega)$, hypotheses \textnormal{(\ref{j1})} and \textnormal{(\ref{j2})} are satisfied,
\begin{align}
\lim_{T\to 0}& \frac{T^{\max(\frac12-H, 0)}}{\sqrt{T}}\mathbb{E} \left(\int_0^T( v_s^2(T-s))^{-3} \Lambda_s\left( \int_s^T \Lambda_r dr \right)  ds \right) = 0,\label{j3}\\
\lim_{T\to 0}& \frac{T^{\max(\frac12-H, 0)}}{\sqrt{T}}\mathbb{E} \left(\int_0^T( v_s^2(T-s))^{-2} \phi_s\left( \int_s^T D^W_s\Lambda_r dr \right)  ds \right)= 0,\label{j4}
\end{align}
and
\begin{equation} \label{j5}
\lim_{T\to 0} \frac{T^{\max(\frac12-H, 0)}}{T^2}\mathbb{E}\left(\frac{1}{v_0^3} \int_0^T \Lambda_s ds \right) <\infty.
\end{equation}
Then,
$$
  \lim_{T\to 0} T^{\max(\frac12-H, 0)}\partial_k I(0,k^*)=\frac12 \lim_{T\to 0}\frac{T^{\max(\frac12-H, 0)}}{T^2}\mathbb{E}\left(\frac{1}{v_0^3} \int_0^T \Lambda_s ds \right).
$$
\end{theorem}

\begin{remark}
Observe that there are two typo in Theorem 8 of Al\`os et al. \cite{Alos2018}. First a factor $T^{-\gamma}$ missing in their hypothesis {\bf (H5)}. Here we are taking $\gamma=\min(H-\frac12, 0) \in (-\frac12, 0]$. Moreover  there is a square missing in the  $u_s(T-s)$ and it should be  $u_s^2(T-s)$, see for example Lemma 6.3.1 in \cite{Alos2021a}. 
\end{remark}

We next present some  technical lemmas that will be  needed in order to check that the hypotheses of the preceding theorems are satisfied.
\begin{lemma} \label{add}
Assume Hypothesis 1. Then, for every $p \geq 1$, there exists a constant $c_{p}>0$ such that
for all $0\leq t <T \leq 1$,
$$
\left(\mathbb{E}\left[v_t^{-2p}\right]\right)^{1/p}
 \leq  c_p \frac{T^{2}}{(T-t)^{2}}.
$$
\end{lemma}

\begin{proof}
We follow a similar idea used in Lemma 3 of \cite{Alos2019}.
We observe that by the definition of $\phi_t$ and equation (\ref{er}), we get that
$$
\int_t^T \phi_r^2 dr=\int_t^T  \left(\frac{\sigma_r S_r (T-r)}{\int_0^r S_u du +S_r (T-r)}\right)^2.
$$
Then, using Hypothesis \ref{Hyp1} we get that
\begin{equation*} \begin{split}
\int_t^T \phi_r^2 dr \geq c_1^2 \exp \left(-4 \sup_{t \in [0,T]} \bigg\vert \int_0^t \sigma_s dW_s-\frac12 \int_0^t \sigma^2_s ds\bigg\vert\right) \frac{(T-t)^3}{3T^2}.
\end{split}
\end{equation*}
Thus, using again Hypothesis 1,
\begin{equation} \label{poc}\begin{split}
\left(\int_t^T \phi_r^2 dr \right)^{-p} \leq c_1^{-2p} e^{2pT c_2^2}\exp \left(4p \sup_{t \in [0,T]} \bigg\vert \int_0^t \sigma_s dW_s\bigg\vert \right) \frac{3^p T^{2p}}{(T-t)^{3p}}.
\end{split}
\end{equation}
By Burkholder-Davis-Gundy inequality and Hypothesis 1, for  any integer $n \geq 1$,
$$
\mathbb{E} \left(\sup_{t \in [0,T]} \bigg\vert \int_0^t \sigma_s dW_s\bigg\vert^n \right) \leq  C n^{n/2} (c T)^{n/2},
$$ 
for some positive constants $c,C$.
Therefore, for $T \leq 1$,
\begin{equation} \label{stir}
\mathbb{E}\exp \left(4p \sup_{t \in [0,T]} \bigg\vert \int_0^t \sigma_s dW_s\bigg\vert \right)
\leq C \sum_{n=1}^{\infty} \frac{(c n)^{n/2}}{n!},
\end{equation}
which is a convergent series. This completes the proof.
\end{proof}

\begin{lemma} \label{l3}
Assume Hypothesis 1. Then, for any $p \geq 1$ there exists a constant $c_p>0$ such  that for all $0\leq t \leq T \leq 1$,
$$
 \E ( M_t^{-p}) \leq c_p.
$$
\end{lemma}

\begin{proof}
Using (\ref{er}) and a similar argument as in the proof of Lemma  \ref{add}, we get  that
$$
\E(M_t^{-p} )\leq e^{p T c_2^2} \E \exp 
\left( p \sup_{t \in [0,T]} \bigg\vert \int_0^t \sigma_s dW_s\bigg\vert \right),
$$
and the result follows from (\ref{stir}).
\end{proof}

Next, we obtain approximation formulas for $\phi$ and its Malliavin derivative.
\begin{lemma} \label{key approx}
Under Hypotheses 1,2, \textnormal{(\ref{d1})}, and 4 the following holds for all $0 \leq s \leq r \leq T$,
\begin{align}
\phi_r&= \frac{\sigma_0(T-r)}{T}+ X^1_{r}, \label{ss1}\\
\phi^2_r&=\frac{\sigma^2_0(T-r)^2}{T^2}+ X^2_{r}, \label{ss2}\\
D_s^W\phi_r&= \frac{\rho(T-r) D_s^{W'}\sigma_r}{T}+\frac{(T-r)\sigma_0^2}{T}-\frac{(T-r)(T-s)\sigma_0^2}{T^2}+X^3_{T,r,s},\label{dphi}\\
D_s^W\phi_r^2&= \frac{2\sigma_0\rho(T-r)^2 D_s^{W'}\sigma_r}{T^2}+\frac{2(T-r)^2 \sigma_0^3}{T^2}
-\frac{2(T-r)^2 (T-s)\sigma_0^3}{T^3}+X^4_{r,s}, \label{dphi2}
\end{align}
where $X^i$ are random variables satisfying for all $0 \leq s \leq r \leq T\leq 1$,
\begin{align}
\vert  X^1_{r}\vert &\leq Y^1_{r} \frac{(T-r) r^{\gamma}}{T},\nonumber \\
\vert  X^2_{r}\vert &\leq Y^2_r \frac{(T-r)^2 r^{\gamma}}{T^2}, \nonumber\\
\vert  X^3_{r,s}\vert &\leq Y^3_{r,s} \frac{(T-r) (r-s)^{H}}{T},\nonumber\\
\vert  X^4_{r,s}\vert &\leq Y^4_{r,s} \frac{(T-r)^2 r^{\gamma} (r-s)^H}{T^2},\nonumber
\end{align}
and $Y^i$ are positive random variables satisfying for all $p \geq 1$
$$
\E( \sup_{0 \leq s \leq r \leq T\leq 1} \vert  Y^i_{r,s}\vert^p )\leq c_i
$$
for some positive constants $c_i$ only dependent on $p$ and $\gamma>0$ is from Hypothesis 4.
\end{lemma}
\(Proof.\) 
We start proving the decomposition for \(\phi_r\).  We  consider  the function
$$
F(S_s, M_s) := \frac{\sigma_0 S_s(T-r)}{T M_s},\quad 0\leq s \leq r.
$$
Observe that
$$
\phi_r=F(S_r, M_r)+(\sigma_r-\sigma_0)\frac{S_r(T-r)}{T M_r}.
$$
Then, using  It\^o's lemma, we get  that
we get
\begin{align*}
F(S_r, M_r)=  \frac{\sigma_0(T-r)}{T}&+\frac{(T-r)}{T} \bigg\{ \int_0^r \frac{\sigma_0}{M_s} dS_s\\
&-\int_0^r \frac{ \sigma_0 S_s}{M^2_s} dM_s+\int_0^r \frac{\sigma_0 \sigma_s^2 S_s^3}{M_s^3} \frac{(T-s)^2}{T^2}ds\bigg\}.
\end{align*}
Then, using Hypotheses \ref{Hyp1} and \ref{Hyp4} and Lemma \ref{l3}, we conclude (\ref{ss1}).
Similarly, we can write 
$$
\phi_r^2=F^2(S_r, M_r)+(\sigma_r^2-\sigma_0^2)\frac{S_r^2(T-r)^2}{T^2 M_r^2}.
$$
and applying It\^o's formula to the function $F^2(S_s, M_s)$  to obtain (\ref{ss2}).

We next prove (\ref{dphi}). 
Using expression (\ref{a2}) of  the Malliavin derivatives computed in the Appendix, we see that the leading terms are equal to
\begin{equation*} \begin{split}
&\frac{\rho (T-r) S_r D_s^{W'}\sigma_r}{T M_r}+ \frac{\rho^2 (T-r) \sigma_r S_r \sigma_s}{T M_r}- \frac{\rho^2(T-r)S_r\sigma_r \sigma_s S_s (T-s)}{T^2 M_r^2} \\
&\quad + \frac{(1-\rho^2) (T-r) \sigma_r S_r \sigma_s}{T M_r} - \frac{(1-\rho^2)(T-r)S_r\sigma_r \sigma_s S_s (T-s)}{T^2 M_r^2}\\=
&\frac{\rho (T-r) S_r D_s^{W'}\sigma_r}{T M_r}+ \frac{(T-r) \sigma_r S_r \sigma_s}{T M_r}- \frac{(T-r)S_r\sigma_r \sigma_s S_s (T-s)}{T^2 M_r^2}.
\end{split}
\end{equation*}
Then, applying It\^o's formula to the functions $F(S_s, M_s)=\frac{S_s}{M_s}$ and $F(S_s, M_s)=\frac{S_s^2 }{M_s^2}$ as above, we obtain (\ref{dphi}).

Finally, in order to check (\ref{dphi2}) it suffice to use the formula
$
D_s^W\phi_r^2=2\phi_r D_s^W\phi_r
$
together with (\ref{ss1}) and (\ref{dphi}). This concludes the proof.
\(\hfill\fbox{\phantom{\rule{.1ex}{.1ex}}}\)

\section{Proof of Theorem \ref{limskew}}\label{limiting-behaviour-of-the-level-and-the-skew-of-the-atm-implied-volatility-of-an-asian-call-option}

\subsection{Proof of (\ref{main1}) in Theorem \ref{limskew}: ATM implied volatility
level}\label{limiting-behaviour-of-the-at-the-money-implied-volatility-level}

  By (\ref{limit}), it suffices to check that the Hypotheses of Theorem \ref{theorem8A} hold true. It is easy to check that Hypotheses 1 and 2 imply the first two hypotheses of Theorem \ref{theorem8A} (see Remarks \ref{r1} and \ref{r2}).

We next check (\ref{j1}). 
Using equation \eqref{poc} and Cauchy-Schwarz inequality we get that
\begin{equation*}
\begin{split}
 \mathbb{E}\int_0^T\frac{\Lambda_s}{v_s^2(T-s)}ds &\leq \int_0^T \frac{T^2}{(T-s)^3} \mathbb{E}(X_T \vert \Lambda_s \vert)ds \\
 &\leq \int_0^T \frac{T^2}{(T-s)^3} \left(\mathbb{E}(X_T^2)\right)^{1/2} \left(\mathbb{E}(\Lambda_s^2)\right)^{1/2}ds,
\end{split}
\end{equation*}
where $X_T = 3 c_1^{-2} e^{2T c_2^2}\exp \left(4 \sup_{t \in [0,T]} \bigg\vert \int_0^t \sigma_s dW_s\bigg\vert \right)$.

Next, due to equation \eqref{stir} we conclude that $\left(\mathbb{E}(X_T^2)\right)^{1/2}$ is bounded as $T\leq 1$. Since $\phi_t \leq \sigma_t$, Cauchy-Schwarz inequality, Lemma \ref{key approx} and Hypothesis (\ref{d1}) imply that
\begin{equation*}
\begin{split}
\mathbb{E}(\Lambda_s^2) &\leq C (T-s)\int_s^T \mathbb{E}(\vert D^W_s\phi_r^2\vert^2) dr \\
 &=O \left((T-s)\int_s^T 
\frac{(T-r)^4}{T^4} (r-s)^{2H-1} dr\right)\\
&= O \left(\frac{(T-s)^{5+2H}}{T^4}\right).
\end{split}
\end{equation*}
Finally, we conclude  that
\begin{equation*}
\begin{split}
 \mathbb{E}\int_0^T\frac{\Lambda_s}{v_s^2(T-s)}ds = O \left( \int_0^T (T-s)^{H-\frac{1}{2}} ds \right) = O \left( T^{H+\frac12}\right),
\end{split}
\end{equation*}
which proves (\ref{j1}).

Similarly, in order to check (\ref{j2}), 
we use Lemma \ref{key approx} together with Cauchy-Schwarz inequality, to get that
\begin{equation*} 
\begin{split}
&\frac{1}{T^2}\mathbb{E}\left( \frac{1}{v_0} \int_0^T \left(\int_s^T D^W_s\phi_r^2 dr\right)^2 ds\right)  \leq \frac{C}{T^2} \int_0^T (T-s) \int_s^T \mathbb{E}\left(\left(D^W_s\phi_r^2\right)^2\right) drds\\
& \qquad =O\left( \frac{1}{T^2} \int_0^T (T-s) \int_s^T \frac{(T-r)^4}{T^4} (r-s)^{2H-1} dr ds\right)\\
&\qquad =O\left( \int_0^T \frac{(T-s)^{5+2H}}{T^6} ds \right)\\
&\qquad=O(T^{2H}).
\end{split}
\end{equation*}

Thus, condition \eqref{j2} also holds and the proof is completed.

\(\hfill\fbox{\phantom{\rule{.1ex}{.1ex}}}\)

\subsection{Proof of (\ref{main2}) in Theorem \ref{limskew}: ATM implied volatility
skew}\label{limiting-behaviour-of-the-at-the-money-implied-volatility-skew}

We will apply Theorem \ref{theorem8B}.
We start checking hypothesis (\ref{j3}). 

Using (\ref{poc}), Lemma \ref{key approx} and Hypothesis (\ref{d1}), we get that
\begin{equation*} 
\begin{split}
&\E\left( \int_0^T \left(\int_s^T \phi_r^2 dr\right)^{-3} \Lambda_s\left( \int_s^T \Lambda_r dr \right)  ds \right)  \\
&=O\bigg(\E\bigg(\int_0^T \frac{X_T T^6}{(T-s)^9}\frac{T-s}{T}\int_s^t \frac{(T-r)^2}{T^2} \vert D_s^{W'} \sigma_r \vert dr \\
& \qquad \qquad \qquad \times 
\int_s^t \frac{T-r}{T} \int_r^T \frac{(T-u)^2}{T^2} \vert D_r^{W'} \sigma_u \vert du  dr\bigg) \bigg)\\
&=O\left(\int_0^T \frac{1}{(T-s)^3}\int_s^t (r-s)^{H-\frac12} dr
\int_s^t \int_r^T  (u-r)^{H-\frac12} du  dr\right)\\
&=O\left(\int_0^T (T-s)^{2H-1}\right)=O(T^{2H}).
\end{split}
\end{equation*}
Thus, as $\lim_{T \rightarrow 0}\frac{T^{\max(\frac12-H, 0)}}{\sqrt{T}} T^{2H}=\lim_{T \rightarrow 0} T^{\max(H, 2H-\frac12)}=0$ for all $H \in (0,1)$,  we get that (\ref{j3}) holds true for the leading terms of Lemma \ref{key approx}. The other terms can be treated similarly and we conclude that (\ref{j3}) holds true.

We next check (\ref{j4}). 
By the definition of $\Lambda_s$, we have \begin{align*}
\begin{split}
\int_s^TD_s^W\Lambda_rdr &= \int_s^TD_s^W\left(\phi_r\int_r^TD_r^W\phi_u^2du\right)dr\\
&= \int_s^T\left(\left(D_s^W\phi_r\right)\int_r^TD_r^W\phi_u^2du+\phi_r\int_r^TD_s^WD_r^W\phi_u^2du\right)dr,
\end{split}
\end{align*}
where
\begin{align*}
\begin{split}
D_s^WD_r^W\phi_u^2&=2(D_s^W\phi_uD_r^W\phi_u+\phi_uD_s^WD_r^W\phi_u).
\end{split}
\end{align*}
Next, using Lemma \ref{key approx} and Hypothesis \ref{Hyp3}, we get that the leading terms are in expectation of order
\begin{align*}
\begin{split}
& \int_s^T D_s^W\phi_r\int_r^TD_r^W\phi_u^2dudr \\
 & = O\left(\int_s^T \frac{(T-r)(r-s)^{H-\frac12}}{T}\int_r^T \frac{(T-u)^2 (u-r)^{H-\frac12}}{T^2}dudr \right)\\
&=O\left(\frac{(T-s)^{4+2H}}{T^3} \right),
\end{split}
\end{align*}
and
\begin{align*}
\begin{split}
& \int_s^T D_s^W\phi_r\int_r^TD_r^W\phi_u^2dudr \\
  & = O\left(\int_s^T \frac{(T-r)}{T}\int_r^T \frac{(T-u)^2 (u-r)^{H-\frac12}(u-s)^{H-\frac12}}{T^2}dudr \right)\\
&=O\left(\frac{(T-s)^{4+2H}}{T^3} \right).
\end{split}
\end{align*}
Then, we conclude that the leading terms satisfy that
\begin{equation*} 
\begin{split}
&\mathbb{E} \left(\int_0^T\left( \int_s^T \phi_r^2 dr\right)^{-2} \phi_s\left( \int_s^T D^W_s\Lambda_r dr \right)  ds \right) \\
&= O\left(\int_0^T \frac{(T^3}{(T-s)^5}\frac{(T-s)^{4+2H}}{T^3} ds\right)\\
&=O\left(\int_0^T (T-s)^{2H-1}\right)=O(T^{2H}).
\end{split}
\end{equation*}
Following as above this proves (\ref{j4}).

We are left to check hypothesis (\ref{j5}). Similarly as above, we have that
$$
\frac{T^{\max(\frac12-H, 0)}}{T^2}\mathbb{E}\left(\frac{1}{v_0^3} \int_0^T \Lambda_s ds \right)=O(1),
$$
and thus the limit is finite. Therefore, all the hypotheses of Theorem \ref{theorem8B} are satisfied.

We finally compute the limit of (\ref{j5}) to check that it coincides with (\ref{main2}). Using Lemma \ref{key approx}, we obtain 
\begin{equation*} \begin{split}
&\lim_{T\to 0} \frac{1}{T^2}\mathbb{E}\left(\frac{T^{\max(\frac12-H, 0)}}{v_0^3} \int_0^T \phi_s \left(\int_s^T D^W_s\phi_r^2dr\right)^2 ds \right)\\
&=\lim_{T \to 0}\E \bigg(\frac{T^{\max(\frac12-H, 0)}}{T^2 v_0^3} \int_0^T  \frac{\sigma_0 (T-s)}{T}
\int_s^T \bigg(\frac{2\sigma_0\rho(T-r)^2 D_s^{W'}\sigma_r}{T^2}\\
&\qquad  +\frac{2(T-r)^2\sigma_0^3}{T^2} -\frac{2(T-r)^2 (T-s)\sigma_0^3}{T^3}\bigg) dr ds\bigg)\\
& =\lim_{T \to 0}\E \left(\frac{2\sigma_0^2 \rho T^{\max(\frac12-H, 0)}}{T^5 v_0^3} \int_0^T  \left( (T-s)
\int_s^T (T-r)^2 D_s^{W'}\sigma_r dr\right) ds\right)\\
&\qquad +\lim_{T \to 0}T^{\max(\frac12-H, 0)}\E \left(\frac{\sigma_0^4 }{45 v_0^3}\right).
\end{split}
\end{equation*}
Using (\ref{poc}) and dominated convergence we see that   
$$\lim_{T \to 0}T^{\max(\frac12-H, 0)}\E \left(\frac{\sigma_0^4 }{45 v_0^3}\right)=\lim_{T \to 0}T^{\max(\frac12-H, 0)} \frac{\sqrt{3} \sigma_0}{60},$$ since $v_0^2$ converges a.s. towards $\frac{\sigma_0^2}{3}$ as $T \rightarrow 0$. 
In order to compute the  remaining limit we write
\begin{equation*} \begin{split}
&\lim_{T \to 0}\E \left(\frac{2\sigma_0^2 \rho T^{\max(\frac12-H, 0)}}{T^5 v_0^3} \int_0^T  \left( (T-s)
\int_s^T (T-r)^2 D_s^{W'}\sigma_r dr\right) ds\right)\\
&\qquad =\lim_{T \to 0}\E \left(\left(\frac{1}{v_0^3}-
\frac{3\sqrt{3}}{\sigma_0^3} \right) A_T \right) +\lim_{T \to 0}\frac{3\sqrt{3}}{\sigma_0^3}\E \left(A_T\right),  
\end{split}
\end{equation*}
where
$$
A_T=\frac{2\rho \sigma_0^2 T^{\max(\frac12-H, 0)}}{T^5}\int_0^T  \left( (T-s)
\int_s^T (T-r)^2 D_s^{W'}\sigma_r dr\right) ds.
$$
By dominated convergence we see  that 
$$
\lim_{T \to 0}\E \left(\left(\frac{1}{v_0^3}-
\frac{3\sqrt{3}}{\sigma_0^3} \right) A_T \right) =0,
$$
which
concludes the proof of (\ref{main2}).
\(\hfill\fbox{\phantom{\rule{.1ex}{.1ex}}}\)

\section{Numerical analysis}\label{numerical-analysis}

In this section we present numerical evidence of the
adequacy of Theorems \ref{limskew} in different settings. 

\subsection{The Black-Scholes model under constant volatility}\label{the-black-scholes-model}

We consider the Black-Scholes model (\ref{bspm}) under constant volatility $\sigma>0$, that is, 
$$
  dS_t = \sigma S_tdW_t, \qquad S_t=S_0e^{\sigma W_t-\frac{\sigma^2}{2}t}.
$$
Appealing to Theorem \ref{limskew} with $\rho=0$ and $H=\frac12$, we conclude that the level and the skew of the
at-the-money implied volatility satisfy  that 
$$
\lim_{T\to 0} I(0,k^{*}) =\frac{\sigma}{\sqrt{3}} \quad \text{ and } \quad \lim_{T \to 0} \partial_k I(0,k^{*}) = \frac{\sigma \sqrt{3}}{30}.
$$
Notice that these results coincide with the ones obtained in Pirjol and
Zhu \cite{Pirjol2016}, see Section 5.4 below.

We next proceed with numerical simulations with  parameters 
$$
S_0=10,\qquad  T=\frac{1}{252},\qquad \sigma \in [0.1, 0.2, \dots, 1.4].
$$

We use the control variates method in order to get estimates of an Asian call option price. As a control variate we use a geometric Asian call option whose price is given by
\begin{equation} \label{geom}
BS_{GeomAsian} = e^{-\frac{1}{4}\sigma_G^2 T}S_0 N(d_1)-KN(d_2),
\end{equation}
where
$$
d_1 = \frac{\log\frac{S_0}{K} + \frac{1}{4} \sigma_G^2T }{\sigma_G \sqrt{T}},\quad  d_2 = d_1 - \sigma_G \sqrt{T}, \quad \sigma_G = \frac{\sigma}{\sqrt{3}}.
$$
Then, the Asian call option price estimator has the following form 
\begin{equation} \label{controlvariates}
\hat{BS}_{Asian} = \frac{1}{N}\sum_{i=1}^N V_T^i - c^*\frac{1}{N}\sum_{i=1}^N(\hat{BS}_{Asian}^i-BS_{GeomAsian}),
\end{equation}
where
$$
c^* =  \frac{\sum_{i=1}^N(V_T^i-\frac{1}{N}\sum_{i=1}^N A_T^i)(\hat{BS}_{Asian}^i-BS_{GeomAsian})}{\sum_{i=1}^N(\hat{BS}_{Asian}^i-BS_{GeomAsian})^2},
$$
and
$$
\hat{BS}_{Asian}^i = \max(\sqrt{S_0^i S_1^i \dots S_m^i} -K,0),
$$
where $N=2000000$, $m=50$, $V_T^i=\max(A_T^i-K,0)$ and the sub-index $i$ indicates the quantity estimated from a realisation of a path from Monte Carlo simulation.

In order to retrieve an estimation for the implied volatility $\hat{I}(0, k^{\ast})$ from the estimated Asian call price we use the algorithm presented in J\"ackel \cite{Jackel2015a}. 
For the estimation of the skew, we use the following finite difference approximation 
\begin{align}
\begin{split}
\partial_k \hat{I}(0,k^{*}) = \frac{\hat{I}(0,k^{*}\log(1+\Delta k)) - \hat{I}(0,\frac{k^{*}}{\log(1+\Delta  k)})}{2\log(1+ \Delta k)},
\end{split}
\label{skew_estimator}
\end{align}
where $\Delta k = 0.001$.

The at-the-money level and the skew of the implied volatility are
presented at Figure \ref{fig1}. We conclude that the results of the numerical simulation are in line with the presented theoretical formulas.

\begin{figure}[h]
\centering
\begin{tabular}{ccc}
  \includegraphics[width=60mm]{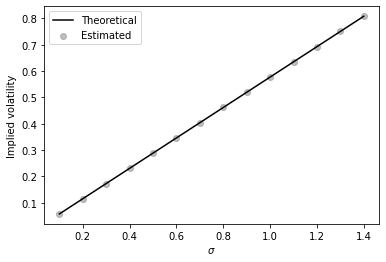} &   \includegraphics[width=60mm]{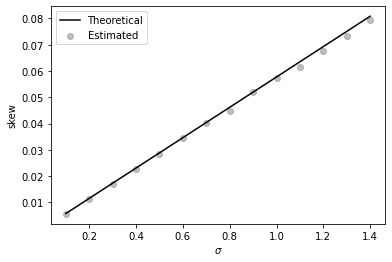} \\[2pt]
\end{tabular}
\caption{At-the-money level and skew of the IV under  Black-Scholes}
\label{fig1}
\end{figure}

\subsection{The SABR model}\label{the-sabr-model}

In this section we consider the SABR stochastic volatility model with skewness parameter 1, which is the most common case from a practical point of view. This corresponds to equation (\ref{bspm}), where $S_t$ denotes the  forward price of the underlying asset and
$$
    d\sigma_t  = \alpha\sigma_t dW_t', \qquad \sigma_t=\sigma_0e^{\alpha W'_t-\frac{\alpha^2}{2}t}.
  $$ 
  where $\alpha>0$ is the volatility of volatility.

Notice that this model does not satisfy Hypothesis 1, so a truncation argument similar as in Section 5 in Al\`os and Shiraya \cite{Alos2019b} is needed in order to check that Theorem \ref{limskew} is true for this model.
We define $\varphi(x)=\sigma_0 \exp(x)$.  For every  $n>1$, we consider a function $\varphi_n \in C^2_b$  satisfying that $\varphi_n(x)=\varphi(x)$ for any $x \in [-n,n]$, 
$\varphi_n(x) \in [\varphi(-2n) \vee \varphi(x), \varphi(-n)]$ for $x \leq -n$, and $\varphi_n(x) \in  [\varphi(n), \varphi(x) \wedge \varphi(2n)]$
for $x \geq n$. We set
$$
\sigma_t^n=\varphi_n\left(\alpha W'_t-\frac{\alpha^2}{2}t\right).
$$
It is easy to see that  $\sigma_t^n$ satisfies Hypotheses 1, 2, (\ref{d1}), and 4. In  fact, for \(r\leq t\), we have that
$$
D_r^{W'}\sigma^n_t = \varphi_n'\left(\alpha W'_t-\frac{\alpha^2}{2}t\right)\alpha,
$$
which implies that  (\ref{d1}) holds with $H=\frac12$ and Hypothesis 4 is satisfied with $\gamma<1/2$. Therefore, appealing to Theorem \ref{limskew} and using the fact that $\sigma_0^n=\sigma_0$, we conclude that
\begin{equation} \label{limit1}
\lim_{T\to 0}I^n(0,k^{*})=\frac{\sigma_0}{\sqrt{3}}.
\end{equation}
where $I^n$ denotes the implied volatility under the volatility process $\sigma^n_t$. We then write
$$
I(0,k^{*})=I^n(0,k^{*})+I(0,k^{*})-I^n(0,k^{*}).
$$
By the mean value theorem,
\begin{equation*} \begin{split}
I(0,k^{*})-I^n(0,k^{*})&=\partial_\sigma (BS^{-1}(0,X_0,X_0,\xi))(V_0-V_0^n) \\
&=e^{-X_0+\frac{\xi^2 T}{8}}\frac{\sqrt{2\pi}}{\sqrt{T}} (V_0-V_0^n),
\end{split}
\end{equation*}
for some $\xi\in (V_0, V_0^n)$, where $V_0^n$ is the option price under $\sigma^n$.
Thus, for $T \leq 1$ and $n>\alpha^2$,
\begin{eqnarray*}
\vert I(0,k^{*})-I^n(0,k^{*}) \vert &\leq & \frac{C_n}{\sqrt{T}}\E\left (|e^{X_T}-e^{X^n_T}|{\bf1}_{\sup_{s\in [0,T]} \vert \ln (\sigma_s/\sigma_0) \vert >n}\right)\nonumber\\
&\leq & \frac{C_n}{\sqrt{T}}\E [(|e^{X_T}+e^{X^n_T}|^2)]^{1/2} \left[\P \left(\sup_{s\in [0,T]}\vert \ln (\sigma_s/\sigma_0) \vert>n \right)\right]^{1/2}\nonumber\\
&\leq& \frac{C_n}{\sqrt{T}}\left[\P\left(\sup_{s\in [0,T]}|\alpha W_s'-\alpha^2s/2|>n\right)\right]^{\frac12}\nonumber\\
&\leq& \frac{C_n}{\sqrt{T}}\left[\P\left(\sup_{s\in [0,T]}|W_s|>\frac{n}{2\alpha}\right)\right]^{\frac12}\nonumber,\\
\end{eqnarray*}
for some constant $C_n>0$ that changes from line to line.
Then, Markov's inequality 
implies that  for all $p>2$,
$$\vert I(0,k^{*})-I^n(0,k^{*})\vert\leq \frac{C_n}{\sqrt{T}} \left(\frac{2\alpha}{n} \right)^{p/2} \left[\E\left( \sup_{s\in [0,T]}|W_s|^p\right)\right]^{1/2}\leq C_n T^{\frac{p}{2}-\frac12},
$$
 Thus, taking $p>4$, we conclude that $$\lim_{T\to 0} I(0,k^{*}) =\frac{\sigma_0}{\sqrt{3}}.$$

On the other hand, for $s \leq r \leq t$, we have
$$
D_s^{W'} D_r^{W'}\sigma^n_t = \varphi_n^{''}\left(\alpha W'_t-\frac{\alpha^2}{2}t\right)\alpha^2,
$$
which implies that  (\ref{d2}) holds with $H=\frac12$.
Therefore, appealing to Theorem \ref{limskew} we get that
$$
\quad \lim_{T \to 0} \partial_k I^n(0,k^{*}) = \frac{\sqrt{3}\rho \alpha\varphi_n'(\sigma_0)}{5\sigma_0^n} +\frac{\sqrt{3}\sigma_0^n }{30}= \frac{\sqrt{3}\rho \alpha}{5} +\frac{\sqrt{3}\sigma_0 }{30} .
$$
Next, similarly as above we can write
\begin{equation*} \begin{split}
\partial_k I(0,k^{*})=\partial_k I^n(0,k^{*})+\partial_k(I(0,k^{*})-I^n(0,k^{*})).
\end{split}
\end{equation*}
By the mean value theorem,
\begin{equation*} \begin{split}
\partial_k(I(0,k^{*})-I^n(0,k^{*}))&=\partial_\sigma \partial_k (BS^{-1}(0,X_0,X_0,\xi))(V_0-V_0^n) \\
&=-e^{-X_0+\frac{\xi^2 T}{8}}\frac{\sqrt{2\pi}}{2} \xi(V_0-V_0^n),
\end{split}
\end{equation*}
for some $\xi\in (V_0, V_0^n)$. Thus, proceeding as above we conclude that
$$
\quad \lim_{T \to 0} \partial_k I(0,k^{*}) = \frac{\sqrt{3}\rho \alpha}{5} +\frac{\sqrt{3}\sigma_0 }{30} .
$$

We next proceed with some numerical simulations using the following parameters
$$S_0 = 10, \, T=\frac{1}{252}, \, dt=\frac{T}{50}, \, \alpha = 0.5, \, \rho = -0.3, \, \sigma_0 =(0.1, 0.2, \dots, 1.4).$$

In order to get estimates of an Asian call option we use antithetic variates. The estimate of the price is defined as follows
\begin{align}
\begin{split}
\hat{V}_{sabr} &= \frac{ \frac{1}{N}\sum_{i=1}^N V_T^i + \frac{1}{N}\sum_{i=1}^N V_T^{i,A} }{2},
\end{split}
\label{antithetic}
\end{align}
where $N=2000000$ and the sub-index $A$ denotes the value of an Asian call option computed on the antithetic trajectory of a Monte Carlo path.

We use equation \eqref{skew_estimator} in order to get estimates of the skew. In Figure \ref{fig2} we present the results of a Monte Carlo simulation
which aims to evaluate numerically the level and the skew of the at-the-money implied volatility of an Asian call option under the SABR model. Again, the numerical results fit the theoretical ones.

\begin{figure}[h]
\centering
\begin{tabular}{ccc}
  \includegraphics[width=60mm]{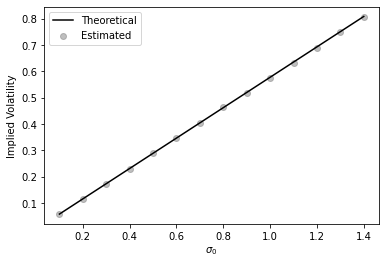} &   \includegraphics[width=60mm]{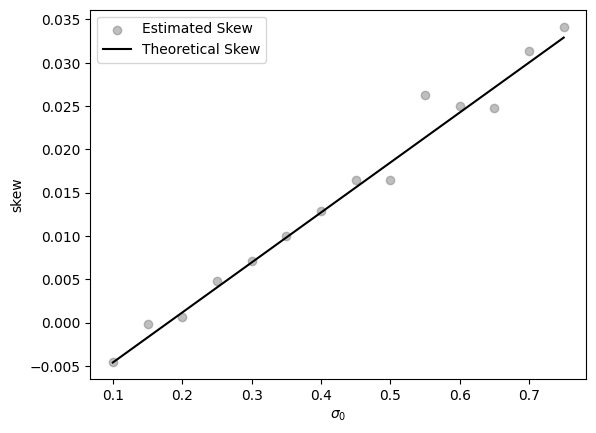} \\[2pt]
\end{tabular}
\caption{At-the-money level and skew of the IV under SABR model.}
\label{fig2}
\end{figure}

Finally, we use the linear approximation \eqref{IVProxy} for higher maturities. We run a Monte Carlo simulation with model parameters $\rho=-0.3, \alpha=0.2$ and $\sigma_0=0.4$. The result is presented in Figure \ref{figSABRIVproxy}. As one can see, taking into account the scale, the asymptotic formula performs well in the case of higher maturities and acts as an upper bound of the implied volatility.  Note that discrepancy increases with the maturity of the option, as expected. 
\begin{figure}[h]
\centering
\begin{tabular}{ccc}
  \includegraphics[width=40mm]{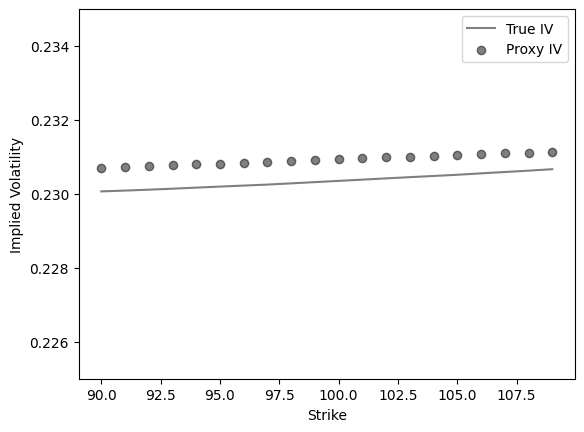} &   \includegraphics[width=40mm]{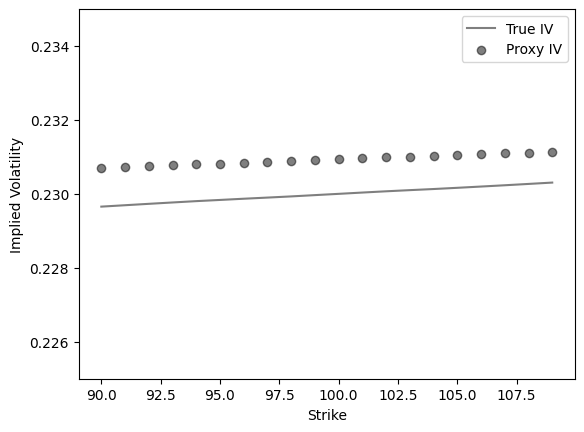} &   \includegraphics[width=40mm]{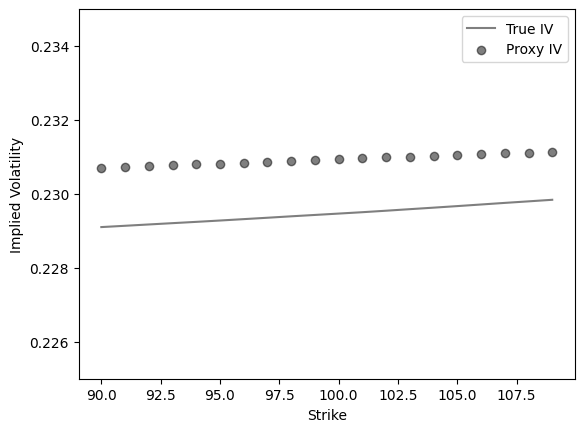}\\
  (a) 3 Months & (b) 6 Months & (b) 1 Year\\[2pt] 
\end{tabular}
\caption{Implied volatility approximation as a function of the strike. The SABR model.}
\label{figSABRIVproxy}
\end{figure}

\subsection{The fractional Bergomi model}\label{rough-bergomi-model}

The fractional Bergomi stochastic volatility model asssumes equation (\ref{bspm}) with
\begin{equation*} \begin{split}
    \sigma_t^2 = \sigma_0^2 e^{v \sqrt{2H}Z_t-\frac{1}{2}v^2t^{2H}}, \quad 
    Z_t = \int_0^t(t-s)^{H-\frac{1}{2}}dW'_s,
    \end{split}
    \end{equation*}
    where $H \in  (0,1)$ and $v>0$, see Example 2.5.1 in Al\`os and Garc\'ia-Lorite \cite{Alos2021a}.
    
    As for the SABR model, a truncation argument is needed in order to apply Theorem \ref{limskew}, as Hypothesis 1 is not satisfied.
We define $\varphi$ and $\varphi_n$ as for the SABR model, and we set
$$
\sigma_t^n=\varphi_n\left(\frac{1}{2} v \sqrt{2H}Z_t-\frac{1}{4}v^2t^{2H}\right).
$$
It is easy to see that  $\sigma_t^n$ satisfies Hypotheses 1, 2, (\ref{d1}), and 4. In fact, for \(r\leq t\), we have that
$$
D_r^{W'}\sigma^n_t = \varphi_n'\left(\frac{1}{2} v \sqrt{2H}Z_t-\frac{1}{4}v^2t^{2H}\right)\frac12 
v \sqrt{2 H} (t-r)^{H-\frac{1}{2}},
$$
which implies that Hypothesis (\ref{d1}) holds  and Hypothesis 4 is satisfied with $\gamma<H$.
Moreover, for $s \leq r \leq t$, we have
$$
D_s^{W'} D_r^{W'}\sigma^n_t = \varphi_n^{''}\left(\frac{1}{2} v \sqrt{2H}Z_t-\frac{1}{4}v^2t^{2H}\right)\frac{1}{4} 
v^2 \sqrt{4 H^4} (t-r)^{H-\frac{1}{2}}(t-s)^{H-\frac{1}{2}},
$$
which implies that  (\ref{d2}) holds.
Therefore, by Theorem \ref{limskew} we get that (\ref{limit1}) holds. Concerning the short maturity limit of the skew, we observe that
$$
\E(D_r^{W'}\sigma_u)=e^{-\frac18 v^2 u^{2H}}\frac12  \sigma_0
v \sqrt{2 H} (u-r)^{H-\frac{1}{2}}.
$$
which gives
\begin{equation} \label{rou}
\lim_{T \to 0} \partial_kI^n(0,k^{*}) 
 =\begin{cases}
\frac{\sqrt{3}\sigma_0 }{30} \quad &\text{if} \quad H>\frac12 \\
\frac{\sqrt{3}\rho v }{10} +\frac{\sqrt{3}\sigma_0 }{30} \quad &\text{if} \quad H=\frac12,
\end{cases}
\end{equation}
and  for $H<\frac12$
\begin{equation} \begin{split}\label{rou2}
&\lim_{T \to 0} T^{\frac{1}{2}-H}\left( \partial_k I^n(0,k^{*})-\frac{\sqrt{3}\sigma_0 }{30}\right)\\
&\qquad \qquad=\frac{3\sqrt{6H}  \rho v}{(1+H-\frac{1}{2}) (2 +H-\frac{1}{2}) (3+H-\frac{1}{2}) (5+H-\frac{1}{2})}.
\end{split}
\end{equation}
Finally, similarly as for the SABR model one can easily show that for $n$ sufficiently large but fixed,
$$
\lim_{T \rightarrow 0} (I(0,k^{*})-I^n(0,k^{*}))=0
$$
and
$$
\lim_{T \rightarrow 0} \partial_k(I(0,k^{*})-I^n(0,k^{*}))=0,
$$
so (\ref{limit1}), (\ref{rou}), and (\ref{rou2}) are also true for $I(0, k^{*})$.

The parameters used for the Monte Carlo simulation are the following
$$
S_0 = 10, \, T=0.001, \, dt=\frac{T}{50}, \, H=(0.4, 0.7), \, v = 0.5, \, \rho = -0.3, \, \sigma_0 =(0.1, 0.2, \dots, 1.4).$$

In order to obtain an estimate of the price of an Asian call option under the fractional Bergomi model we use the combination of antithetic and control variates presented in equations \eqref{controlvariates} and \eqref{antithetic}. That is, we  first sample the process from the  Bergomi model and the antithetic analogue.
We then average the payoffs calculated from both paths.
Finally, use the geometric Asian as control variate assuming constant volatility model at level $\sigma_0$. 

In Figure \ref{fig5} we plot the estimates of the level of the ATMIV of the Asian
call option and we observe that the result is independent of $H$ as stated in Theorem \ref{limskew}.
\begin{figure}[h]
\centering
\begin{tabular}{ccc}
  \includegraphics[width=60mm]{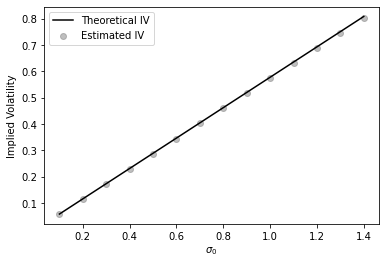} &   \includegraphics[width=60mm]{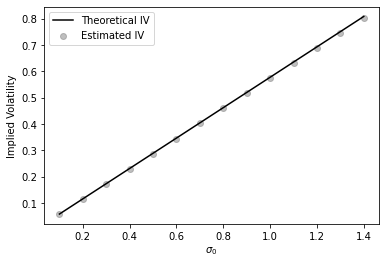} \\
(a) H=0.4 & (b) H=0.7 \\[2pt]
\end{tabular}
\caption{At-the-money level of the IV under fractional Bergomi model}
\label{fig5}
\end{figure}

In Figure \ref{fig6} we simulate the ATMIV skew  of the Asian call option as a function of the maturity as well as its least squares fit in order to observe the blow up to $-\infty$ for the case $H=0.4$.
\begin{figure}[h]
\centering
\begin{tabular}{ccc}
  \includegraphics[width=60mm]{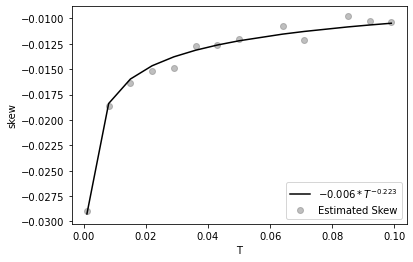} &   \includegraphics[width=60mm]{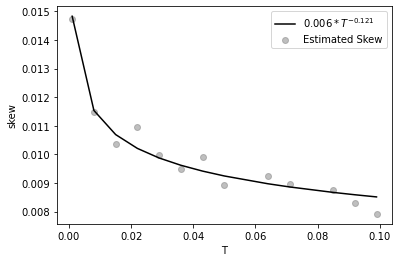} \\
H=0.4 &  H=0.7 \\[2pt]
\end{tabular}
\caption{At-the-money IV skew as a function of $T$ under fractional Bergomi model}
\label{fig6}
\end{figure}

We then plot the quantities $T^{\frac{1}{2}-H}\partial_k \hat{I}(0,k^{*})$ for $H=0.4$ and $\partial_k \hat{I}(0,k^{*})$ for $H=0.7$ in Figure \ref{fig7} as a function of $\sigma_0$.  For $H=0.4$, the line $-0.0243+0.032\sigma_0$ corresponds to the least square fit while formula (\ref{rou2}) gives the line $-0.0286+0.029 \sigma_0$. This difference is due to the  numerical instability of  the finite difference estimation at short maturity in the presence of rough noise and could be improved by increasing considerably the number of Monte Carlo samples or applying  a variance reduction technique.
For $H=0.7$, we observe  that formula (\ref{rou}) fits well the Monte Carlo estimates.
\begin{figure}[h]
\centering
\begin{tabular}{ccc}
  \includegraphics[width=60mm]{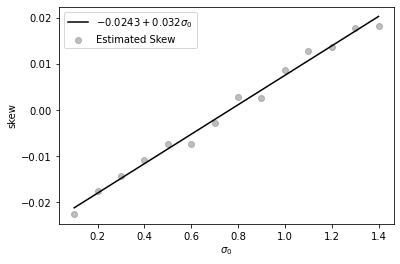} &   \includegraphics[width=60mm]{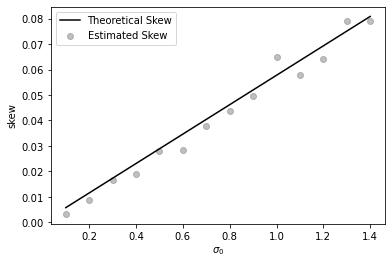} \\
 H=0.4 &  H=0.7 \\[2pt]
\end{tabular}
\caption{At-the-money IV skew as a function of $\sigma_0$ under fractional Bergomi model}
\label{fig7}
\end{figure}

As for the SABR model, we use the approximation \eqref{IVProxy} for higher maturities. We run a Monte Carlo simulation with model parameters $\rho=-0.3, v=0.2, h=0.4$ and $\sigma_0=0.4$. The result is presented in Figure \ref{figRBIVproxy}. We conclude that the approximation performs well since the scale is very small and the approximation acts as an upper bound for an actual implied volatility.
\begin{figure}[h]
\centering
\begin{tabular}{ccc}
  \includegraphics[width=40mm]{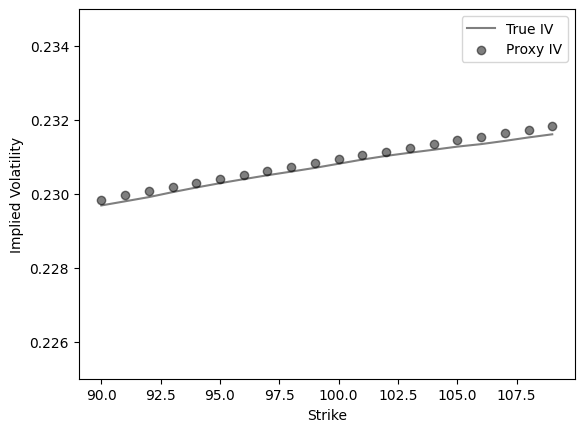} &   \includegraphics[width=40mm]{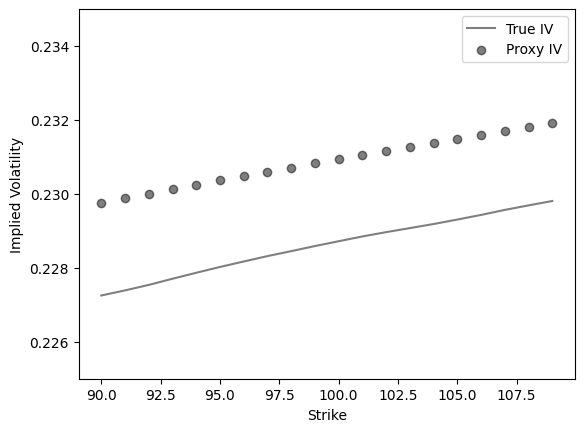} &   \includegraphics[width=40mm]{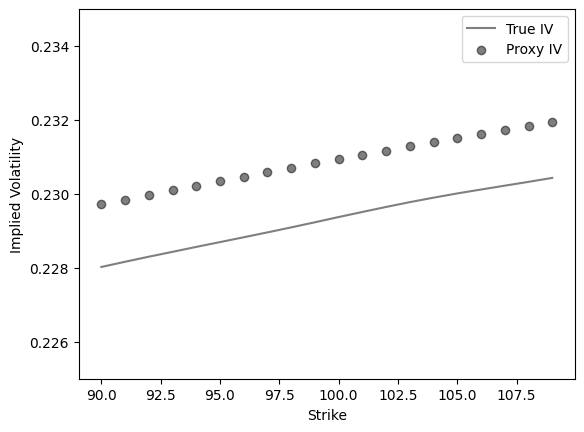}\\
  (a) 3 Months & (b) 6 Months & (b) 1 Year\\[2pt] 
\end{tabular}
\caption{Implied volatility approximation as a function of the strike. Rough Bergomi model.}
\label{figRBIVproxy}
\end{figure}

\subsection{Local volatility model}\label{local-volatility-model}

The short-maturity limit of the ATMIV level and skew of an Asian option under local volatility has already been computed in  Pirjol and Zhu \cite{Pirjol2016}. The aim of this section is 
to check that our Theorem \ref{limskew} provides the same asymptotics as the ones obtained in  that paper.
We  consider the local volatility model
\begin{equation} \label{local}
dS_t=\sigma(S_t) S_t dW_t,
\end{equation}
where \(\sigma(.)\) is a twice
differentiable function. 
In Proposition 19 of Pirjol and Zhu \cite{Pirjol2016} they show that  the following expansion holds for \(x = \log(\frac{K}{S_0})\) around
the ATM point 
 \begin{equation} \label{expan}
\lim_{T \to 0} I(0,k^*)=\frac{\sigma(S_0)}{\sqrt{3}}\left( 1+\left( \frac{1}{10}+\frac{3\sigma'(S_0)}{5\sigma(S_0)}S_0 \right)x + O(x^2) \right).
\end{equation}
Then, differentiating equation (\ref{expan}), we obtain that
\begin{equation} \label{hj}
\lim_{T \to 0} \partial_k I(0,k^{*}) = \frac{1}{\sqrt{3}}\left(\frac{1}{10} \sigma(S_0)+ 
\frac{3}{5}S_0 \sigma'(S_0) \right).
 \end{equation}

We next apply Theorem \ref{limskew} in the case of the  local volatility model (\ref{local}) with  $\sigma_t=\sigma(S_t)$ and $\rho=1$ to verify that  we obtain the same expressions as in (\ref{main1}) and (\ref{hj}). For the level, we directly see that when $\sigma(S_t)$ equals $\sigma_t$ and $K=S_0$, (\ref{expan}) coincides with the limit in (\ref{main1}). 
For the skew, we need to compute $D_r \sigma(S_t)$.  We have for $r \leq u$,
\begin{equation*} \begin{split}
D_r \sigma(S_u)=\sigma'(S_u) D_r (S_u)=\sigma'(S_u) \left(\sigma(S_r) S_r+\int_r^u D_r(\sigma(S_s) S_s) dW_s \right).
\end{split}
\end{equation*}
In particular,
\begin{equation*}
\E\left(D_r \sigma(S_u) \right)=\E\left(\sigma'(S_u)\sigma(S_r) S_r\right).
\end{equation*}
This can be written as 
\begin{equation*} \begin{split}
\E\left(D_r \sigma(S_u) \right)=\sigma'(S_0) \sigma(S_0)S_0 &+\E\left((\sigma'(S_u)-\sigma'(S_0)) \sigma(S_r) S_r\right)\\
& +\sigma'(S_0) \E\left((\sigma(S_r)-\sigma(S_0)) S_r\right).
\end{split}
\end{equation*}
Then,  using the mean value theorem and the fact that $S_t$ has H\"older continuous sample paths of any order $\gamma<\frac12$, we see that the last two terms of the last display will not contribute in the limit (\ref{main2}). Thus, (\ref{main2}) gives
 \begin{align*}
  \begin{split}
\lim_{T \to 0} \partial_kI(0,k^{*})=\frac{\sqrt{3}}{5} S_0 \sigma'(S_0)+\frac{\sqrt{3}\sigma(S_0) }{30},
  \end{split}
\end{align*}
which is the same as in (\ref{hj}). This serves as one more evidence of the validity of
Theorem \ref{limskew}.

\appendix

\section{Computation of Malliavin derivatives}\label{some-malliavin-derivatives}

In this section we provide the computations of the first and second  Malliavin derivatives of the processes \(S_t\),
\(M_t\) and $\phi_t$ defined in Section 2.

Using the fact that $\sigma_t$ is adapted to the filtration of $W'$ and the formula for the derivative of a stochastic integral (see for example (3.6) in Nualart and Nualart \cite{Nua18}), we get that, for $0\leq s \leq r \leq T$,
\begin{align*}
\begin{split}
D_s^{W'}S_r &= S_r\left(\rho \sigma_s-\frac{1}{2}\int_s^rD_s^{W'}\sigma_u^2du +\int_s^rD_s^{W'}\sigma_u dW_u \right),\\
D_s^{B}S_r &= S_r\sigma_s\sqrt{1-\rho^2},\\
D_s^{W'}M_r &= \frac{\rho\sigma_s S_s(T-s)}{T}+\int_s^r \frac{(T-u)D_s^{W'}(\sigma_u S_u)}{T}dW_u,\\
D_s^BM_r &= \frac{\sqrt{1-\rho^2}\sigma_s S_s(T-s)}{T}+\int_s^r \frac{(T-u)\sigma_uD_s^{B}( S_u)}{T}dW_u.
\end{split}
\end{align*}
Moreover, appealing to (\ref{derprod}), we find that
\begin{align*} 
\begin{split}
D_s^WS_r & = \rho S_r\left(-\frac{1}{2}\int_s^rD_s^{W'}\sigma_u^2du +\int_s^rD_s^{W'}\sigma_u dW_u \right) + S_r \sigma_s,\\
D_s^WM_r&=  \frac{\sigma_s S_s(T-s)}{T}+\rho \int_s^r \frac{(T-u)D_s^{W'}(\sigma_u S_u)}{T}dW_u \\
&\qquad \qquad + \sqrt{1-\rho^2}\int_s^r \frac{(T-u)D_s^B(\sigma_u S_u)}{T}dW_u.
\end{split}
\end{align*}
Finally, from the definition of $\phi_t$, we conclude that
\begin{align} \label{a2}
\begin{split}
D_s^W\phi_r&=\frac{\rho (T-r)D_s^{W'}(\sigma_r S_r)}{T M_r} - \frac{\rho(T-r)S_r\sigma_rD_s^{W'}M_r}{T M_r^2} \\
&+ \frac{\sqrt{1-\rho^2} (T-r)D_s^{B}(\sigma_r S_r)}{T M_r} - \frac{\sqrt{1-\rho^2}(T-r)S_r\sigma_rD_s^{B}M_r}{T M_r^2}.
\end{split}
\end{align}

We next compute the second Malliavin derivatives. Similarly as before, using the fact that we can differentiate Lebesgue integrals of stochastic processes (see for example Proposition 3.4.3 in Nualart and Nualart \cite{Nua18}), we get that, for $0\leq s \leq r \leq u \leq T$,
\begin{align}
\begin{split}
D_s^BD_r^{W'}S_u &= S_u \sigma_s \sqrt{1-\rho^2} \left( \rho \sigma_r -\frac{1}{2}\int_r^uD_r^{W'}\sigma_v^2dv + \int_r^uD_r^{W'}\sigma_vdW_v \right),\\
D_s^{W'}D_r^{W'}S_u &= S_u\left(\rho \sigma_s-\frac{1}{2}\int_s^uD_s^{W'}\sigma_v^2dv +\int_s^uD_s^{W'}\sigma_v dW_v \right)\\
&\qquad \qquad\times \left(\rho \sigma_r-\frac{1}{2}\int_r^uD_r^{W'}\sigma_v^2dv +\int_r^uD_r^{W'}\sigma_v dW_v \right) \\
&\qquad +S_u\left( \rho D_s^{W'}\sigma_r-\frac{1}{2}\int_r^uD_s^{W'}D_r^{W'}\sigma_v^2dv +\int_r^uD_s^{W'}D_r^{W'}\sigma_v dW_v \right),\\
D_s^{W'}D_r^BS_u & = \sqrt{1-\rho^2}S_uD_s^{W'}\sigma_r + \sqrt{1-\rho^2}\sigma_rD_s^{W'}S_u,\\
D_s^{W'}D_r^{W'}M_u &= \frac{\rho(T-r)D_s^{W'}(\sigma_r S_r)}{T}+\int_r^u \frac{(T-v)D_s^{W'}D_r^{W'}(\sigma_v S_v)}{T}dW_v,\\
D_s^BD_r^{W'}M_u &= \frac{\rho(T-r)\sigma_rD_s^{B}S_r}{T}+\int_r^u \frac{(T-v)D_s^{B}D_r^{W'}(\sigma_v S_v)}{T}dW_v,\\
D_s^{W'}D_r^BM_u &= \frac{\sqrt{1-\rho^2}(T-r)D_s^{W'}(\sigma_r S_r)}{T}+\int_r^u \frac{(T-v)D_s^{W'}(\sigma_v D_s^{B}S_v)}{T}dW_v,\\
D_s^{B}D_r^BM_u &=\frac{\sqrt{1-\rho^2}(T-r)\sigma_rD_s^{B}( S_r)}{T}+\int_r^u \frac{(T-v)\sigma_v D_s^{B}D_r^{B}S_v}{T}dW_v,
\end{split}
\label{App1}
\end{align}
and
\begin{align}
\begin{split}
&D_s^WD_r^W\phi_u = \frac{\rho^2(T-u)D_s^{W'}D_r^{W'}(\sigma_u S_u)}{T M_u}-\frac{\rho^2(T-u)D_s^{W'}(\sigma_u S_u)D_s^{W'}M_u}{T M_u^2}\\
& -\frac{\rho^2(T-u)D_s^{W'}\left(\sigma_u S_u D_r^{W'}M_u\right)}{T M_u^2} +\frac{2 \rho^2(T-u) \sigma_u S_u D_s^{W'}M_u D_r^{W'}M_u}{T M_u^3}\\
&+\frac{\rho\sqrt{1-\rho^2}(T-u)D_s^{W'}D_r^{B}(\sigma_u S_u)}{T M_u}-\frac{\rho\sqrt{1-\rho^2}(T-u)D_r^{B}(\sigma_u S_u)D_s^{W'}M_u}{T M_u^2}\\
&-\frac{\rho\sqrt{1-\rho^2}(T-u)D_s^{W'}\left(\sigma_u S_u D_r^{B}M_u\right)}{T M_u^2}+\frac{2\rho\sqrt{1-\rho^2}(T-u)\sigma_u S_u D_r^{B}M_uD_s^{W'}M_u}{T M_u^3}\\
&+\frac{\rho\sqrt{1-\rho^2}(T-u)D_s^{B}D_r^{W'}(\sigma_u S_u)}{T M_u}-\frac{\rho\sqrt{1-\rho^2}(T-u)D_r^{W'}(\sigma_u S_u)D_s^{B}M_u}{T M_u^2}\\
&-\frac{\rho\sqrt{1-\rho^2}(T-u)D_s^{B}\left(\sigma_u S_u D_r^{W'}M_u\right)}{T M_u^2}+\frac{\rho\sqrt{1-\rho^2}(T-u)\sigma_u S_uD_r^{W'}M_uD_s^{B}M_u}{T M_u^3}\\
&+\frac{(1-\rho^2)(T-u)D_s^BD_r^B(\sigma_u S_u)}{T M_u}-\frac{(1-\rho^2)(T-u)D_r^B(\sigma_u S_u) D_s^BM_u}{T M_u^2}\\
&-\frac{(1-\rho^2)(T-u)D_s^B\left(S_u \sigma_u D_r^BM_u\right)}{T M_u^2}+\frac{2(1-\rho^2)(T-u)S_u\sigma_uD_r^BM_uD_s^BM_u}{T M_u^3}.
\end{split}
\label{malderL1}
\end{align}

\end{document}